\documentclass[a4paper,UKenglish,cleveref,autoref, thm-restate]{lipics-v2021}

\title{Permutation Match Puzzles: How Young Tanvi Learned About Computational Complexity}

\titlerunning{Permutation Match Puzzles}

\author{Kshitij Gajjar}{IIIT Hyderabad, India}{kshitij@iiit.ac.in}{0000-0003-0890-199X}{}
\author{Neeldhara Misra}{IIT Gandhinagar, India}{neeldhara.m@iitgn.ac.in}{0000-0003-1727-5388}{}

\authorrunning{K. Gajjar and N. Misra}

\Copyright{Kshitij Gajjar and Neeldhara Misra}



\category{} 

\relatedversion{} 

\acknowledgements{We would like to thank the Center for Creative Learning (\href{https://ccl.iitgn.ac.in/}{CCL}) at IIT Gandhinagar for introducing us to this puzzle. We also acknowledge the use of an AI assistant (\href{https://claude.ai/}{Claude}) for the preparation of this draft. We are also grateful to the reviewers of FUN 2026 for their feedback.}

\nolinenumbers 

\EventEditors{John Iacono}
\EventNoEds{1}
\EventLongTitle{13th International Conference on Fun with Algorithms (FUN 2026)}
\EventShortTitle{FUN 2026}
\EventAcronym{FUN}
\EventYear{2026}
\EventDate{May 18--22, 2026}
\EventLocation{Porquerolles, France}
\EventLogo{}
\SeriesVolume{366}
\ArticleNo{38}

\usepackage{tikz}
\usetikzlibrary{arrows,shapes,positioning,calc}
\usetikzlibrary{arrows.meta, decorations.markings}
\newenvironment{mydefinition}[2]
  {\vspace{0.5\baselineskip}\begin{definition}[#1]\label{#2}}
  {\end{definition}}

\let\oldtheorem\theorem
\let\endoldtheorem\endtheorem
\renewenvironment{theorem}{\vspace{0.5\baselineskip}\oldtheorem}{\endoldtheorem}

\let\oldlemma\lemma
\let\endoldlemma\endlemma
\renewenvironment{lemma}{\vspace{0.5\baselineskip}\oldlemma}{\endoldlemma}

\let\oldobservation\observation
\let\endoldobservation\endobservation
\renewenvironment{observation}{\vspace{0.5\baselineskip}\oldobservation}{\endoldobservation}

\let\oldcorollary\corollary
\let\endoldcorollary\endcorollary
\renewenvironment{corollary}{\vspace{0.5\baselineskip}\oldcorollary}{\endoldcorollary}

\let\oldproposition\proposition

\renewenvironment{proposition}{\vspace{0.5\baselineskip}\oldproposition}{\endproposition}
\usepackage{eulervm,framed}

\begin{document}

\maketitle

\begin{abstract}
We study a family of sorting match puzzles on grids, which we call permutation match puzzles. In this puzzle, each row and column of a $n \times n$ grid is labeled with an ordering constraint $\text{---}$ ascending (A) or descending (D) $\text{---}$ and the goal is to fill the grid with the numbers 1 through $n^2$ such that each row and column respects its constraint.

We provide a complete characterization of solvable puzzles: a puzzle admits a solution if and only if its associated constraint graph is acyclic, which translates to a simple ``at most one switch'' condition on the A/D labels. When solutions exist, we show that their count is given by a hook length formula. For unsolvable puzzles, we present an $O(n)$ algorithm to compute the minimum number of label flips required to reach a solvable configuration. Finally, we consider a generalization where rows and columns may specify arbitrary permutations rather than simple orderings, and establish that finding minimal repairs in this setting is NP-complete by a reduction from feedback arc set.
\end{abstract}

    \begin{figure}
    \centering
    \includegraphics[width=0.45\textwidth]{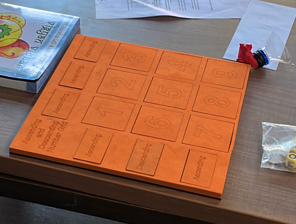}
    \caption{A 3 × 3 sorting match puzzle: the physical edition that inspired this work}
    \label{fig:puzzle-photo}
    \end{figure}
    
\section{Introduction}

Tanvi, while wandering the corridors of IIT Gandhinagar, ambles into the Centre for Creative Learning (CCL), where she is intrigued by a really simple looking puzzle --- also shown in~\Cref{fig:puzzle-photo}. From the looks of it, it is a $3 \times 3$ grid on a foam board with nine removable pieces labeled from $1$ to $9$. It reminds her of a magic square, but here each row and column is labelled by the word \emph{ascending} or \emph{descending} --- it turns out these are also removable bits that can be flipped over to change the label, so the board appears to hold sixty-four puzzles in all! 

She wonders what the puzzle demands. After asking around, she finds that the goal is to reorganize the numbers in the grid such that each row and column is in the appropriate order. She wonders if there is a way to solve each of the sixty-four puzzles in a systematic way. A playable version of this puzzle can be found \href{https://interactives.neeldhara.com/puzzles/asc-desc-grid}{here}.

In this work, we explore general versions of this puzzle, providing a fairly complete picture of questions of solvability, counting solutions, and the issue of optimally perturbing unsolvable puzzles to make them solvable.

Consider the following natural generalization of the $3 \times 3$ puzzle found at CCL: given an $n \times n$ grid where each row and column is labeled either A (ascending) or D (descending), fill the grid with numbers from 1 to $n^2$ such that each row (respectively, column) labeled A must contain numbers in ascending order and each row (respectively, column) labeled D must contain numbers in descending order.
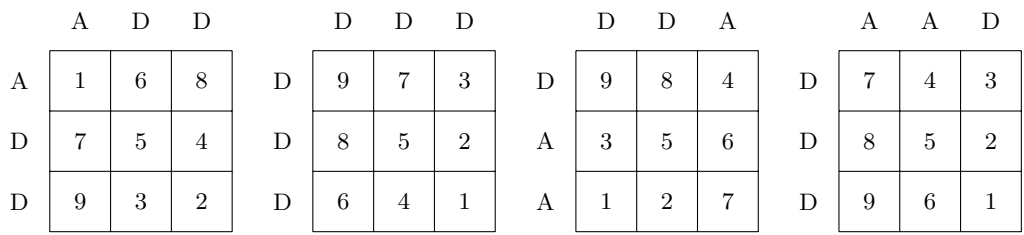
\begin{figure}[h]
    \centering
    \begin{tabular}{cccc}
    \begin{tikzpicture}[scale=0.8]
    \draw (0,0) grid (3,3);
    \node at (-0.5,2.5) {A}; \node at (-0.5,1.5) {D}; \node at (-0.5,0.5) {D};
    \node at (0.5,3.5) {A}; \node at (1.5,3.5) {D}; \node at (2.5,3.5) {D};
    \node at (0.5,2.5) {1}; \node at (1.5,2.5) {6}; \node at (2.5,2.5) {8};
    \node at (0.5,1.5) {7}; \node at (1.5,1.5) {5}; \node at (2.5,1.5) {4};
    \node at (0.5,0.5) {9}; \node at (1.5,0.5) {3}; \node at (2.5,0.5) {2};
    \end{tikzpicture}
    &
    \begin{tikzpicture}[scale=0.8]
    \draw (0,0) grid (3,3);
    \node at (-0.5,2.5) {D}; \node at (-0.5,1.5) {D}; \node at (-0.5,0.5) {D};
    \node at (0.5,3.5) {D}; \node at (1.5,3.5) {D}; \node at (2.5,3.5) {D};
    \node at (0.5,2.5) {9}; \node at (1.5,2.5) {7}; \node at (2.5,2.5) {3};
    \node at (0.5,1.5) {8}; \node at (1.5,1.5) {5}; \node at (2.5,1.5) {2};
    \node at (0.5,0.5) {6}; \node at (1.5,0.5) {4}; \node at (2.5,0.5) {1};
    \end{tikzpicture}
    &
    \begin{tikzpicture}[scale=0.8]
    \draw (0,0) grid (3,3);
    \node at (-0.5,2.5) {D}; \node at (-0.5,1.5) {A}; \node at (-0.5,0.5) {A};
    \node at (0.5,3.5) {D}; \node at (1.5,3.5) {D}; \node at (2.5,3.5) {A};
    \node at (0.5,2.5) {9}; \node at (1.5,2.5) {8}; \node at (2.5,2.5) {4};
    \node at (0.5,1.5) {3}; \node at (1.5,1.5) {5}; \node at (2.5,1.5) {6};
    \node at (0.5,0.5) {1}; \node at (1.5,0.5) {2}; \node at (2.5,0.5) {7};
    \end{tikzpicture}
    &
    \begin{tikzpicture}[scale=0.8]
    \draw (0,0) grid (3,3);
    \node at (-0.5,2.5) {D}; \node at (-0.5,1.5) {D}; \node at (-0.5,0.5) {D};
    \node at (0.5,3.5) {A}; \node at (1.5,3.5) {A}; \node at (2.5,3.5) {D};
    \node at (0.5,2.5) {7}; \node at (1.5,2.5) {4}; \node at (2.5,2.5) {3};
    \node at (0.5,1.5) {8}; \node at (1.5,1.5) {5}; \node at (2.5,1.5) {2};
    \node at (0.5,0.5) {9}; \node at (1.5,0.5) {6}; \node at (2.5,0.5) {1};
    \end{tikzpicture}
    \end{tabular}
    \caption{Examples of 3 × 3 ascending-descending grid puzzles with valid solutions.}
\end{figure}

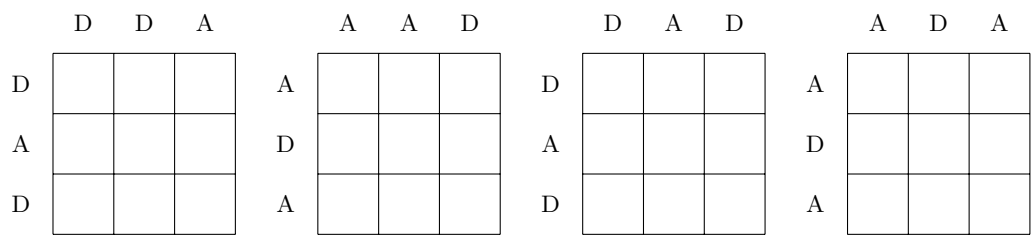
\begin{figure}[h]
\centering
\begin{tabular}{cccc}
\begin{tikzpicture}[scale=0.8]
\draw (0,0) grid (3,3);
\node[left] at (-0.23,2.5) {D};
\node[left] at (-0.23,1.5) {A};
\node[left] at (-0.23,0.5) {D};
\node[above] at (0.5,3.23) {D};
\node[above] at (1.5,3.23) {D};
\node[above] at (2.5,3.23) {A};
\end{tikzpicture}
&
\begin{tikzpicture}[scale=0.8]
\draw (0,0) grid (3,3);
\node[left] at (-0.23,2.5) {A};
\node[left] at (-0.23,1.5) {D};
\node[left] at (-0.23,0.5) {A};
\node[above] at (0.5,3.23) {A};
\node[above] at (1.5,3.23) {A};
\node[above] at (2.5,3.23) {D};
\end{tikzpicture}
&
\begin{tikzpicture}[scale=0.8]
\draw (0,0) grid (3,3);
\node[left] at (-0.23,2.5) {D};
\node[left] at (-0.23,1.5) {A};
\node[left] at (-0.23,0.5) {D};
\node[above] at (0.5,3.23) {D};
\node[above] at (1.5,3.23) {A};
\node[above] at (2.5,3.23) {D};
\end{tikzpicture}
&
\begin{tikzpicture}[scale=0.8]
\draw (0,0) grid (3,3);
\node[left] at (-0.23,2.5) {A};
\node[left] at (-0.23,1.5) {D};
\node[left] at (-0.23,0.5) {A};
\node[above] at (0.5,3.23) {A};
\node[above] at (1.5,3.23) {D};
\node[above] at (2.5,3.23) {A};
\end{tikzpicture}
\end{tabular}
\caption{Four examples of 3 × 3 ascending-descending grid puzzles with no solution.}
\label{fig:some-unsolvable-cases}
\end{figure}

The examples in~\Cref{fig:some-unsolvable-cases} show configurations that have no valid solutions. To understand why these puzzles are unsolvable, we first observe a key structural property that creates an impossible situation: if there exist rows $i,j$ with $i < j$ labeled D and A (or A and D) respectively, and columns $p,q$ with $p < q$ labeled A and D (or D and A) respectively, then we have a circular constraint, illustrated in Figure 3.

\begin{figure}[h]
\centering
\begin{tikzpicture}[scale=1.2]
\draw[dotted] (-0.8,-0.8) rectangle (3.5,2.5);

\draw (0,0) grid (2,2);

\node[left] at (-0.3,1.5) {D};
\node[left] at (-0.3,0.5) {A};
\node[above] at (0.5,2.1) {A};
\node[above] at (1.5,2.1) {D};

\node[right] at (2.3,1.5) {row $i$};
\node[right] at (2.3,0.5) {row $j$};
\node[below] at (0.5,-0.3) {col $p$};
\node[below] at (1.5,-0.3) {col $q$};

\node at (0.5,1.5) {$a$};
\node at (1.5,1.5) {$b$};
\node at (0.5,0.5) {$c$};
\node at (1.5,0.5) {$d$};

\draw[->, thick, purple] (0.7,1.5) -- (1.3,1.5);

\draw[->, thick, purple] (1.5,1.3) -- (1.5,0.7);

\draw[->, thick, purple] (1.3,0.5) -- (0.7,0.5);

\draw[->, thick, purple] (0.5,0.7) -- (0.5,1.3);

\begin{scope}[xshift=5cm]
\draw[dotted] (-0.8,-0.8) rectangle (3.5,2.5);

\draw (0,0) grid (2,2);

\node[left] at (-0.3,1.5) {A};
\node[left] at (-0.3,0.5) {D};
\node[above] at (0.5,2.1) {D};
\node[above] at (1.5,2.1) {A};

\node[right] at (2.3,1.5) {row $i$};
\node[right] at (2.3,0.5) {row $j$};
\node[below] at (0.5,-0.3) {col $p$};
\node[below] at (1.5,-0.3) {col $q$};

\node at (0.5,1.5) {$a$};
\node at (1.5,1.5) {$b$};
\node at (0.5,0.5) {$c$};
\node at (1.5,0.5) {$d$};

\draw[->, thick, purple] (0.5,1.3) -- (0.5,0.7);

\draw[->, thick, purple] (0.7,0.5) -- (1.3,0.5);

\draw[->, thick, purple] (1.5,0.7) -- (1.5,1.3);

\draw[->, thick, purple] (1.3,1.5) -- (0.7,1.5);

\end{scope}

\end{tikzpicture}
\caption{Two illustrations of circular constraints in unsolvable configurations. Left: rows labeled D,A with columns A,D, creating the cycle $a < b < d < c < a$. Right: rows labeled A,D with columns D,A, creating the cycle $a < c < d < b < a$. Both cycles are impossible to satisfy with distinct values.}\label{fig:unsolvable}
\end{figure}
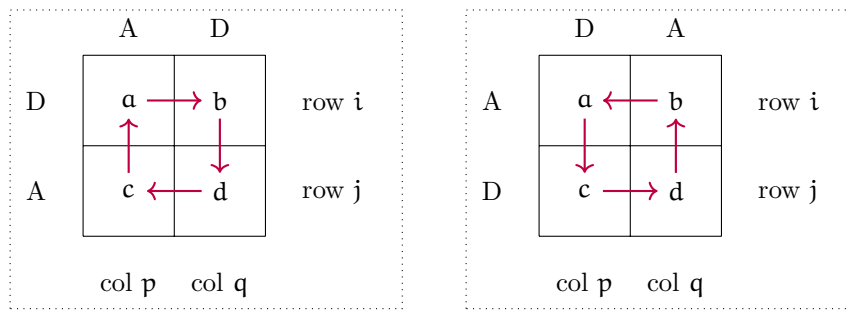

Consider the four cells at positions $(i,p)$, $(i,q)$, $(j,p)$, and $(j,q)$, containing values $a$, $b$, $c$, and $d$ respectively. Following the arrows in the diagram, the constraints force a cyclic chain of inequalities: $a < b < d < c < a$ or $a < c < d < b < a$, depending on the labeling, which is clearly impossible.

We now formally define the notion of a \emph{sorting match puzzle}.


\begin{mydefinition}{Sorting Match Puzzles}{def:general-permutation-puzzle}
A sorting match puzzle of order $n$ is a pair of words $(r,c)$, each of length $n$, over the alphabet $\{A,D\}$. A sorting match puzzle $(r,c)$ is \emph{solvable} if there exists a permutation of the integers $1,2,\ldots,n^2$ arranged in an $n \times n$ grid such that:
    \begin{itemize}
        \item For each $i \in \{1,\ldots,n\}$, if $r_i = A$, the entries in the $i^{th}$ row are in ascending order; if $r_i = D$, the entries are in descending order.
        \item For each $j \in \{1,\ldots,n\}$, if $c_j = A$, the $j^{th}$ column is in ascending order; if $c_j = D$, the entries are in descending order.
    \end{itemize}
\end{mydefinition}

We also call a sorting match puzzle row (respectively, column) uniform if all its row (resp., column) labels are identical.

\begin{mydefinition}{Uniform Sorting Match puzzles}{def:uniform-permutation-puzzle}
A sorting match puzzle $(r,c)$ is row-uniform if $r = A^n$ or $r = D^n$; a sorting match puzzle $(r,c)$ is col-uniform if $c = A^n$ or $c = D^n$; and a sorting match puzzle $(r,c)$ is uniform if it is both row-uniform and col-uniform.
\end{mydefinition}

A first natural question we want to address is the following: when is a sorting match puzzle solvable? Our first result is that the scenarios described in~\Cref{fig:unsolvable} are the \emph{only} obstructions to solvability. In particular, solvable puzzles are either row or column uniform, or are of the form $r = A^k D^{n-k}$ and $c = A^\ell D^{n-\ell}$ for some $k, \ell \in \{1,\ldots,n-1\}$ or $r = D^k A^{n-k}$ and $c = D^\ell A^{n-\ell}$ for some $k, \ell \in \{1,\ldots,n-1\}$. This characterization also tells us that there are $4 \cdot (2^n) + 2 \cdot (n-1)^2$ solvable puzzles. 

Next, we want to explore both solvable and unsolvable puzzles further. For solvable puzzles, we would like to count the number of solutions. For unsolvable puzzles, we would like to know the minimum number of label flips required to make the puzzle solvable. We discover that the answer to the first question involves the hook length formula, while for the second question notice that by guessing the ``shape'' of the nearest solvable puzzle, we can find the answer in $O(n^2)$ time by brute force. We improve this to a linear time algorithm using a simple dynamic programming approach. We note that our ideas here extend to $d$-dimensional grids too.

Notice that ascending and descending orders can be thought of as two special permutations over $[n]$. This suggests the following natural generalization of the puzzle above: where the labels can be arbitrary permutations over $[n]$, and a solution is a permutation of the integers $1,2,\ldots,n^2$ that satisfies the ordering constraints imposed by the labels --- see~\Cref{fig:generalized-example} for an illustration. This motivates our next definition. 

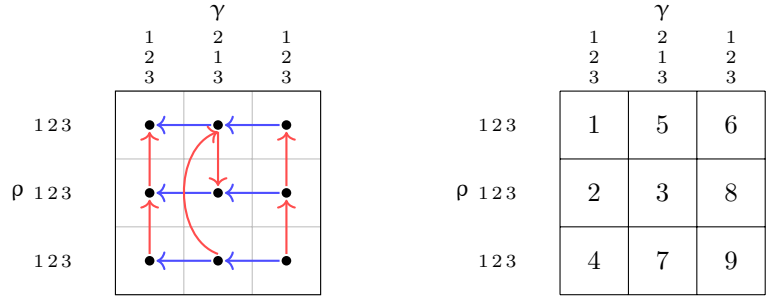
\begin{figure}[h]
\centering
\begin{tikzpicture}[scale=0.9]
\begin{scope}[shift={(0,0)}]
\draw[gray!50] (0,0) grid (3,3);
\draw[black] (0,0) rectangle (3,3);

\node[left] at (-0.5,2.5) {\scriptsize $1\,2\,3$};
\node[left] at (-0.5,1.5) {\scriptsize $1\,2\,3$};
\node[left] at (-0.5,0.5) {\scriptsize $1\,2\,3$};
\node[left] at (-1.2,1.5) {\small $\rho$};

\node at (0.5,3.8) {\scriptsize $1$};
\node at (0.5,3.5) {\scriptsize $2$};
\node at (0.5,3.2) {\scriptsize $3$};

\node at (1.5,3.8) {\scriptsize $2$};
\node at (1.5,3.5) {\scriptsize $1$};
\node at (1.5,3.2) {\scriptsize $3$};

\node at (2.5,3.8) {\scriptsize $1$};
\node at (2.5,3.5) {\scriptsize $2$};
\node at (2.5,3.2) {\scriptsize $3$};

\node at (1.5,4.15) {\small $\gamma$};

\foreach \j in {0,1,2} {
    \foreach \i in {0,1,2} {
        \fill[black] (\i+0.5,\j+0.5) circle (2pt);
    }
}

\draw[->, thick, blue!70] (2.4,2.5) -- (1.6,2.5);
\draw[->, thick, blue!70] (1.4,2.5) -- (0.6,2.5);
\draw[->, thick, blue!70] (2.4,1.5) -- (1.6,1.5);
\draw[->, thick, blue!70] (1.4,1.5) -- (0.6,1.5);
\draw[->, thick, blue!70] (2.4,0.5) -- (1.6,0.5);
\draw[->, thick, blue!70] (1.4,0.5) -- (0.6,0.5);

\draw[->, thick, red!70] (0.5,0.6) -- (0.5,1.4);
\draw[->, thick, red!70] (0.5,1.6) -- (0.5,2.4);

\draw[->, thick, red!70] (1.5,0.6) to[out=160,in=200] (1.5,2.4);
\draw[->, thick, red!70] (1.5,2.4) -- (1.5,1.6);

\draw[->, thick, red!70] (2.5,0.6) -- (2.5,1.4);
\draw[->, thick, red!70] (2.5,1.6) -- (2.5,2.4);

\end{scope}

\begin{scope}[shift={(6.5,0)}]
\draw (0,0) grid (3,3);

\node[left] at (-0.5,2.5) {\scriptsize $1\,2\,3$};
\node[left] at (-0.5,1.5) {\scriptsize $1\,2\,3$};
\node[left] at (-0.5,0.5) {\scriptsize $1\,2\,3$};
\node[left] at (-1.2,1.5) {\small $\rho$};

\node at (0.5,3.8) {\scriptsize $1$};
\node at (0.5,3.5) {\scriptsize $2$};
\node at (0.5,3.2) {\scriptsize $3$};

\node at (1.5,3.8) {\scriptsize $2$};
\node at (1.5,3.5) {\scriptsize $1$};
\node at (1.5,3.2) {\scriptsize $3$};

\node at (2.5,3.8) {\scriptsize $1$};
\node at (2.5,3.5) {\scriptsize $2$};
\node at (2.5,3.2) {\scriptsize $3$};

\node at (1.5,4.15) {\small $\gamma$};

\node at (0.5,2.5) {1};
\node at (1.5,2.5) {5};
\node at (2.5,2.5) {6};
\node at (0.5,1.5) {2};
\node at (1.5,1.5) {3};
\node at (2.5,1.5) {8};
\node at (0.5,0.5) {4};
\node at (1.5,0.5) {7};
\node at (2.5,0.5) {9};

\end{scope}

\end{tikzpicture}
\caption{A generalized sorting match puzzle with row permutations $\rho_1 = \rho_2 = \rho_3 = (1,2,3)$ (all ascending) and column permutations $\gamma_1 = \gamma_3 = (1,2,3)$ (ascending) but $\gamma_2 = (2,1,3)$. Left: The constraint graph where an arrow $P \to Q$ indicates $P > Q$ (blue for rows, red for columns); arrows implied by transitivity are omitted for clarity. For column 2, the permutation $(2,1,3)$ requires: middle entry $<$ top entry $<$ bottom entry. Right: A valid solution where column 2 has $3 < 5 < 7$.}
\label{fig:generalized-example}
\end{figure}

\begin{mydefinition}{Permutation Match Puzzles}{def:generalized-permutation-puzzle}
A \emph{permutation match puzzle} of order $n$ is a pair $(\rho, \gamma)$, where $\rho = (\rho_1, \ldots, \rho_n)$ and $\gamma = (\gamma_1, \ldots, \gamma_n)$ are sequences of permutations, with each $\rho_i, \gamma_j \in S_n$. A generalized sorting match puzzle $(\rho, \gamma)$ is \emph{solvable} if there exists a permutation of the integers $1, 2, \ldots, n^2$ arranged in an $n \times n$ grid such that:
    \begin{itemize}
        \item For each $i \in \{1, \ldots, n\}$, if the entries in the $i^{th}$ row are $a_1, a_2, \ldots, a_n$ (from left to right), then $a_{\rho_i^{-1}(1)} < a_{\rho_i^{-1}(2)} < \cdots < a_{\rho_i^{-1}(n)}$.
        \item For each $j \in \{1, \ldots, n\}$, if the entries in the $j^{th}$ column are $b_1, b_2, \ldots, b_n$ (from top to bottom), then $b_{\gamma_j^{-1}(1)} < b_{\gamma_j^{-1}(2)} < \cdots < b_{\gamma_j^{-1}(n)}$.
    \end{itemize}
\end{mydefinition}

Sorting match puzzles are a special case of permutation match puzzles where each $\rho_i$ and $\gamma_j$ is either the identity permutation (corresponding to A) or the reversal permutation (corresponding to D). 

By interpreting the permutation constraints as directed edges on a grid, we observe that determining if a permutation match puzzle is solvable or not reduces to cycle detection. Permutation match puzzles inspire the following question: given a directed graph $G$ on the vertex set $\{(i,j) \mid 1 \leq i,j \leq n\}$, where the subgraph induced by the vertices $R_i := \{(i,j) \mid 1 \leq j \leq n\}$ and $C_j := \{(i,j) \mid 1 \leq i \leq n\}$ is acyclic for each $i,j \in [n]$, what is the minimum number of edges we need to flip to make the $G$ acyclic?  This is essentially a feedback arc set problem on a restricted class of directed graphs. Our final result is that this problem is NP-complete.

\textbf{Related Work.} Permutation match puzzles belong to a rich tradition of combinatorial puzzles with algebraic and algorithmic underpinnings. The most classical example is the 15-puzzle, whose solvability was characterized by Johnson and Story~\cite{johnson1879notes} using permutation parity: a configuration is reachable if and only if it corresponds to an even permutation. While checking solvability is polynomial-time, Ratner and Warmuth~\cite{ratner1990n} showed that finding the shortest solution for generalized $n \times n$ sliding puzzles is NP-hard.

Our counting results connect to the theory of posets and Young tableaux. The problem of counting linear extensions of a poset is \#P-complete in general~\cite{brightwell1991counting}, making our closed-form formula for sorting match puzzles particularly notable. The formula we derive is reminiscent of the celebrated hook length formula of Frame, Robinson, and Thrall~\cite{frame1954hook}, which counts standard Young tableaux of a given shape. Greene, Nijenhuis, and Wilf~\cite{greene1979probabilistic} gave an elegant probabilistic proof of this formula using a ``hook walk'' argument, which bears conceptual similarity to our recursive decomposition approach.

The structural constraints in permutation match puzzles also relate to geometric grid classes of permutations studied by Albert, Atkinson, Bouvel, Ru\v{s}kuc, and Vatter~\cite{albert2013geometric}, where permutations are constrained by monotonicity conditions on grid cells. Grid-based constraint satisfaction more broadly connects to Latin squares and Sudoku puzzles~\cite{colbourn2010latin, yato2003complexity}, though our puzzles differ in using global ordering constraints rather than local distinctness conditions.

Our NP-completeness result for the minimum repair problem relates to feedback arc set problems. While minimum feedback arc set is NP-hard even for tournaments~\cite{charbit2007minimum}, and approximation algorithms exist for general instances~\cite{ailon2008aggregating}, our reduction shows hardness persists even when the underlying graph has the restricted structure arising from grid-based permutation constraints.

\section{Sorting Match Puzzles}

We first establish a characterization of solvable sorting match puzzles.

\begin{restatable}{theorem}{solvablepuzzlecharacterization}
    \label{thm:solvable-puzzle-characterization}
A sorting match puzzle $(r,c)$ is solvable if and only if it is either row-uniform or col-uniform, or satisfies one of the following conditions:
\begin{enumerate}
    \item There exists $k \in \{1,\ldots,n-1\}$ and $\ell \in \{1,\ldots,n-1\}$ such that:
    \begin{itemize}
        \item The first $k$ rows are labeled A and the remaining rows are labeled D, and
        \item The first $\ell$ columns are labeled A and the remaining columns are labeled D. 
    \end{itemize}
    \item There exists $k \in \{1,\ldots,n-1\}$ and $\ell \in \{1,\ldots,n-1\}$ such that:
    \begin{itemize}
        \item The first $k$ rows are labeled D and the remaining rows are labeled A, and
        \item The first $\ell$ columns are labeled D and the remaining columns are labeled A.
    \end{itemize}
\end{enumerate}
\end{restatable}

\begin{proof}

If $(r,c)$ is col-uniform, say with $c = A^n$, then fill the grid from top-left to bottom-right, row by row, with the entries in the $i^{th}$ row being the numbers $\{(i-1)n + 1, \ldots, (i-1)n + n\}$ in ascending order if $r_i = A$ and in descending order if $r_i = D$. If $c = D^n$, then fill the grid from bottom-right to top-left, row by row, with the entries in the $i^{th}$ row from the bottom being the numbers $\{(i-1)n + 1, \ldots, (i-1)n + n\}$ in ascending order if $r_i = A$ and in descending order if $r_i = D$. An analogous reasoning applies if $(r,c)$ is row-uniform.

Now consider the case when there exists $k \in \{1,\ldots,n-1\}$ and $\ell \in \{1,\ldots,n-1\}$ such that the first $k$ rows are labeled A and the remaining rows are labeled D, and the first $\ell$ columns are labeled A and the remaining columns are labeled D. Let:

\begin{itemize}
\item $P^{\nwarrow}$ denote the cells included in the intersection of the first $k$ rows and the first $\ell$ columns 
\item $P^{\nearrow}$ denote the cells included in the intersection of the first $k$ rows and the remaining columns
\item $P^{\swarrow}$ denote the cells included in the intersection of the remaining rows and the first $\ell$ columns
\item $P^{\searrow}$ denote the cells included in the intersection of the remaining rows and the remaining columns
\end{itemize}

\begin{figure}[ht]
    \centering
    \begin{tikzpicture}
    \begin{scope}[shift={(0,0)},scale=0.7]
        \draw (0,0) grid (3,3);
        \node at (-0.5,2.5) {A};
        \node at (-0.5,1.5) {A};
        \node at (-0.5,0.5) {A};
        \node at (0.5,3.5) {D};
        \node at (1.5,3.5) {A};
        \node at (2.5,3.5) {D};
        \node at (0.5,2.5) {3};
        \node at (1.5,2.5) {4};
        \node at (2.5,2.5) {9};
        \node at (0.5,1.5) {2};
        \node at (1.5,1.5) {5};
        \node at (2.5,1.5) {8};
        \node at (0.5,0.5) {1};
        \node at (1.5,0.5) {6};
        \node at (2.5,0.5) {7};
    \end{scope}
    
    \begin{scope}[shift={(3.75,0)},scale=0.7]
        \draw (0,0) grid (3,3);
        \node at (-0.5,2.5) {D};
        \node at (-0.5,1.5) {D};
        \node at (-0.5,0.5) {D};
        \node at (0.5,3.5) {D};
        \node at (1.5,3.5) {A};
        \node at (2.5,3.5) {D};
        \node at (0.5,2.5) {9};
        \node at (1.5,2.5) {4};
        \node at (2.5,2.5) {3};
        \node at (0.5,1.5) {8};
        \node at (1.5,1.5) {5};
        \node at (2.5,1.5) {2};
        \node at (0.5,0.5) {7};
        \node at (1.5,0.5) {6};
        \node at (2.5,0.5) {1};
    \end{scope}
    
    \begin{scope}[shift={(7.5,0)},scale=0.7]
        \draw (0,0) grid (3,3);
        \node at (-0.5,2.5) {D};
        \node at (-0.5,1.5) {A};
        \node at (-0.5,0.5) {D};
        \node at (0.5,3.5) {A};
        \node at (1.5,3.5) {A};
        \node at (2.5,3.5) {A};
        \node at (0.5,2.5) {3};
        \node at (1.5,2.5) {2};
        \node at (2.5,2.5) {1};
        \node at (0.5,1.5) {4};
        \node at (1.5,1.5) {5};
        \node at (2.5,1.5) {6};
        \node at (0.5,0.5) {9};
        \node at (1.5,0.5) {8};
        \node at (2.5,0.5) {7};
    \end{scope}
    
    \begin{scope}[shift={(11.25,0)},scale=0.7]
        \draw (0,0) grid (3,3);
        \node at (-0.5,2.5) {D};
        \node at (-0.5,1.5) {A};
        \node at (-0.5,0.5) {D};
        \node at (0.5,3.5) {D};
        \node at (1.5,3.5) {D};
        \node at (2.5,3.5) {D};
        \node at (0.5,2.5) {9};
        \node at (1.5,2.5) {8};
        \node at (2.5,2.5) {7};
        \node at (0.5,1.5) {4};
        \node at (1.5,1.5) {5};
        \node at (2.5,1.5) {6};
        \node at (0.5,0.5) {3};
        \node at (1.5,0.5) {2};
        \node at (2.5,0.5) {1};
    \end{scope}
    
    \end{tikzpicture}
    \caption{Examples of row-uniform and column-uniform puzzles with non-uniform labels DAD in the other direction. All examples show valid solutions constructed according to the method described in the proof.}
    \label{fig:uniform-examples}
    \end{figure}

    \begin{figure}[ht]
        \centering
        \begin{tikzpicture}[scale=0.9]
        \definecolor{lightpurple}{RGB}{240,230,255}
        \definecolor{lightblue}{RGB}{230,240,255}
        \definecolor{lightpink}{RGB}{255,230,240}
        \definecolor{lightyellow}{RGB}{255,255,230}
        
        \fill[lightpurple] (0,4) rectangle (3,9);
        \fill[lightblue] (3,4) rectangle (9,9);
        \fill[lightpink] (0,0) rectangle (3,4);
        \fill[lightyellow] (3,0) rectangle (9,4);
        
        \draw (0,0) grid (9,9);
        
        \foreach \i in {0,...,8} {
            \ifnum \i < 5
                \node[left] at (-0.5,8.5-\i) {A};
            \else
                \node[left] at (-0.5,8.5-\i) {D};
            \fi
        }
        
        \foreach \i in {0,...,8} {
            \ifnum \i < 3
                \node[above] at (0.5+\i,9.3) {A};
            \else
                \node[above] at (0.5+\i,9.3) {D};
            \fi
        }
        
        \foreach \i in {0,...,4} {
            \foreach \j in {0,...,2} {
                \pgfmathtruncatemacro{\num}{1 + \i*3 + \j}
                \node at (0.5+\j,8.5-\i) {\num};
            }
        }
        
        \foreach \i in {0,...,3} {
            \foreach \j in {0,...,5} {
                \pgfmathtruncatemacro{\num}{39 - (\i*6 + \j)}
                \node at (3.5+\j,3.5-\i) {\num};
            }
        }
        
        \foreach \i in {0,...,4} {
            \foreach \j in {0,...,5} {
                \pgfmathtruncatemacro{\num}{40 + (4-\i)*6 + \j}
                \node at (3.5+\j,8.5-\i) {\num};
            }
        }
        
        \foreach \i in {0,...,3} {
            \foreach \j in {0,...,2} {
                \pgfmathtruncatemacro{\num}{70 + \i*3 + (2-\j)}
                \node at (0.5+\j,3.5-\i) {\num};
            }
        }
        
        \end{tikzpicture}
        \caption{A valid solution for a 9×9 grid with $k=5$ and $\ell=3$. The grid is partitioned into four subgrids: $P^{\nwarrow}$ (1-15), $P^{\nearrow}$ (40-69), $P^{\swarrow}$ (70-81), and $P^{\searrow}$ (39-16).}
        \label{fig:solvablecasetwo}
        \end{figure}

One solution to this puzzle is the following: populate $P^{\nwarrow}$ with the first $k\ell$ numbers in ascending order, filling this sub-grid row by row, top-left to bottom-right, and populate $P^{\searrow}$ with the next $(n-k)(n-\ell)$ numbers in descending order, filling this sub-grid row by row, top-left to bottom-right. Then we fill the cells in $P^{\nearrow}$ with the next $k(n-\ell)$ numbers in descending order, filling this sub-grid row by row, top-left to bottom-right, however keeping the numbers in each row ascending; and finally, the cells in $P^{\swarrow}$ with the next $(n-k)\ell$ numbers in ascending order, filling this sub-grid row by row, top-left to bottom-right, however keeping the numbers in each row descending.

We refer the reader to~\Cref{fig:solvablecasetwo} for an example of this algorithm with $n = 9, k = 5, \ell = 3$. The other case in the theorem has a symmetric argument. 
\end{proof}

\begin{corollary}
There are $$4 \cdot (2^n-1) + 2 \cdot (n-1)^2$$ sorting match puzzles of order $n$.
\end{corollary}
\begin{proof}
By \Cref{thm:solvable-puzzle-characterization}, solvable puzzles are either: (i) row-uniform puzzles (where $r = A^n$ or $r = D^n$), giving $2 \cdot 2^n$ puzzles (two choices for $r$, and $2^n$ choices for $c$), or (ii) column-uniform puzzles (where $c = A^n$ or $c = D^n$), giving $2 \cdot 2^n$ puzzles (two choices for $c$, and $2^n$ choices for $r$), or (iii) non-uniform (of two types, which we will see later).

This sums up to $4\cdot 2^n$. However, we need to make a small correction for over-counting: we have to subtract four to account for the fact that the four puzzles $(r = A^n, c = A^n)$, $(r = A^n, c = D^n)$, $(r = D^n, c = A^n)$, $(r = D^n, c = D^n)$ are double-counted as they are both row-uniform and column-uniform. This makes it $4\cdot 2^n - 4$ uniform puzzles total.

Finally, let us count the non-uniform puzzles. By \Cref{thm:solvable-puzzle-characterization}, we have two types of non-uniform puzzles: (i) $r = A^k D^{n-k}$ and $c = A^\ell D^{n-\ell}$ for some $k, \ell \in \{1,\ldots,n-1\}$, or (ii) $r = D^k A^{n-k}$ and $c = D^\ell A^{n-\ell}$ for some $k, \ell \in \{1,\ldots,n-1\}$. For each type, there are $(n-1)^2$ solvable puzzles (one puzzle for each choice of $k$ and $\ell$). This makes it $2 \cdot (n-1)^2$ non-uniform puzzles total.
\end{proof}

The first few terms of the series are $4,14,36,78,156$. This sequence is not present on the On-line Encyclopedia of Integer Sequences (\href{https://oeis.org/}{OEIS}), and we plan to put it up soon.

\textbf{Graph-Theoretic Formulation.} We observe that this puzzle can be reformulated as a problem of topologically sorting a directed graph. This view is redundant in the light of the characterization above, but will be useful when we consider the general permutation match puzzles later. 

Given an instance of a sorting match puzzle, we can create a vertex for each cell in the grid, and add the following directed edges based on the A/D labels: for each row labeled A (respectively, D): add edges from left to right (respectively, right to left); for each column labeled A (respectively, D): add edges from bottom to top (respectively, top to bottom).

\begin{theorem}
An instance of the ascending-descending grid puzzle has a solution if and only if its corresponding directed graph is acyclic.
\end{theorem}

\begin{proof}
($\Rightarrow$) If there exists a solution, then the numbers in the grid provide a valid topological ordering of the vertices, which implies the graph must be acyclic.

($\Leftarrow$) If the graph is acyclic, any topological ordering of the vertices provides a valid solution when we replace the vertices with numbers 1 through $n^2$ in order.
\end{proof}

Having established solvability, we now turn to counting solutions for sorting match puzzles that are not uniform. 

\subsection{Counting Solutions}

Consider a $(n \times n)$ sorting match puzzle $(r,c)$, where there exist integers $k \in \{1,\ldots,n-1\}$ and $\ell \in \{1,\ldots,n-1\}$ such that:
    \begin{itemize}
        \item The first $k$ rows are labeled A and the remaining rows are labeled D, and
        \item The first $\ell$ columns are labeled A and the remaining columns are labeled D. 
    \end{itemize}

Recall our notation from before: $P^{\nwarrow}$ for the top-left $k \times \ell$ subgrid, $P^{\nearrow}$ for the top-right $k \times (n-\ell)$ subgrid, $P^{\swarrow}$ for the bottom-left $(n-k) \times \ell$ subgrid, and $P^{\searrow}$ for the bottom-right $(n-k) \times (n-\ell)$ subgrid. Our first claim is that in any valid solution, the numbers from $1$ to $k\ell + (n-k)(n-\ell)$ are used only in $P^{\nwarrow}$ and $P^{\searrow}$.

\begin{lemma}
In any valid solution to this puzzle, the numbers from $1$ to $k\ell + (n-k)(n-\ell)$ must appear only in $P^{\nwarrow}$ and $P^{\searrow}$.
\end{lemma}

\begin{proof}
We first argue that any number in $P^{\nwarrow} \cup P^{\searrow}$ must be smaller than any number in $P^{\swarrow} \cup P^{\nearrow}$. Then the claim follows, since the entries of the grid are a permutation of $1$ to $n^2$.

Consider a number $x$ in $P^{\nwarrow} \cup P^{\searrow}$, located in the $i^{th}$ row and $j^{th}$ column; and a number $y$ in $P^{\swarrow} \cup P^{\nearrow}$, located in the $p^{th}$ row and $q^{th}$ column. There are four cases, and we discuss one of them: the others are analogous.

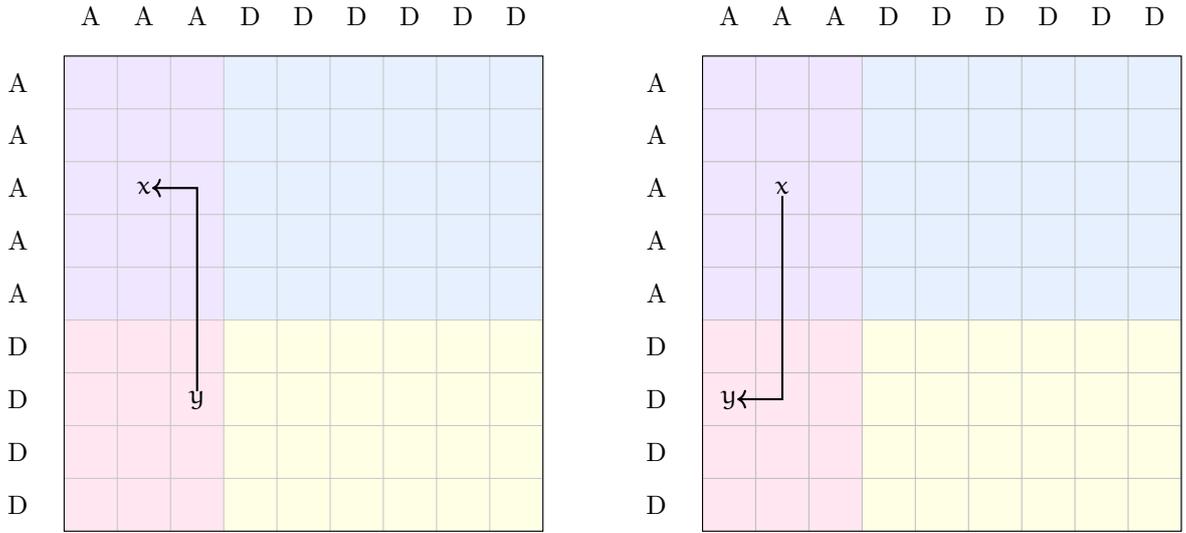
\begin{figure}[h]
    \centering
    \begin{tikzpicture}[scale=0.7]
    \begin{scope}
    \definecolor{lightpurple}{RGB}{240,230,255}
    \definecolor{lightblue}{RGB}{230,240,255}
    \definecolor{lightpink}{RGB}{255,230,240}
    \definecolor{lightyellow}{RGB}{255,255,230}
    \definecolor{lightgray}{RGB}{200,200,200}
    
    \fill[lightpurple] (0,4) rectangle (3,9);
    \fill[lightblue] (3,4) rectangle (9,9);
    \fill[lightpink] (0,0) rectangle (3,4);
    \fill[lightyellow] (3,0) rectangle (9,4);
    
    \draw[lightgray] (0,0) grid (9,9);
    \draw[black] (0,0) rectangle (9,9); 
    
    \foreach \i in {0,...,8} {
        \ifnum \i < 5
            \node[left] at (-0.5,8.5-\i) {A};
        \else
            \node[left] at (-0.5,8.5-\i) {D};
        \fi
    }
    
    \foreach \i in {0,...,8} {
        \ifnum \i < 3
            \node[above] at (0.5+\i,9.4) {A};
        \else
            \node[above] at (0.5+\i,9.4) {D};
        \fi
    }
    
    \node[inner sep=2pt] at (1.5,6.5) {$x$};
    \node[inner sep=2pt] at (2.5,2.5) {$y$};
    \draw[->, black, thick, shorten >=3pt, shorten <=3pt] (2.5,2.5) -- (2.5,6.5) -- (1.5,6.5);
    \end{scope}
    
    \begin{scope}[xshift=12cm]
    \definecolor{lightpurple}{RGB}{240,230,255}
    \definecolor{lightblue}{RGB}{230,240,255}
    \definecolor{lightpink}{RGB}{255,230,240}
    \definecolor{lightyellow}{RGB}{255,255,230}
    
    \fill[lightpurple] (0,4) rectangle (3,9);
    \fill[lightblue] (3,4) rectangle (9,9);
    \fill[lightpink] (0,0) rectangle (3,4);
    \fill[lightyellow] (3,0) rectangle (9,4);
    
    \draw[lightgray] (0,0) grid (9,9);
    \draw[black] (0,0) rectangle (9,9); 
    
    \foreach \i in {0,...,8} {
        \ifnum \i < 5
            \node[left] at (-0.5,8.5-\i) {A};
        \else
            \node[left] at (-0.5,8.5-\i) {D};
        \fi
    }
    
    \foreach \i in {0,...,8} {
        \ifnum \i < 3
            \node[above] at (0.5+\i,9.4) {A};
        \else
            \node[above] at (0.5+\i,9.4) {D};
        \fi
    }
    
    \node[inner sep=2pt] at (1.5,6.5) {$x$};
    \node[inner sep=2pt] at (0.5,2.5) {$y$};
    \draw[->, black, thick, shorten >=3pt, shorten <=3pt] (1.5,6.5) -- (1.5,2.5) -- (0.5,2.5);
    \end{scope}
    
    \end{tikzpicture}
    \caption{Two examples showing paths from $y$ to $x$ in a 9×9 grid with $k=5$ and $\ell=3$. Left: L-shaped path in an ascending column and row. Right: L-shaped path in a descending column and row.}
    \label{fig:counting}
    \end{figure}
    
Suppose $x \in P^{\nwarrow}$ and $y \in P^{\swarrow}$: if $p = q$, then we are in an ascending column and $y > x$. If $p < q$, then $x < z < y$, where $z$ is the number in the $i^{th}$ row and $q^{th}$ column. If $p > q$, then $x < y < z$, where $z$ is the number in the $p^{th}$ row and $j^{th}$ column (see~\Cref{fig:counting} for an illustration).
\end{proof}

An analogous result holds if the first $k$ rows are labeled D and the remaining rows are labeled A, and the first $\ell$ columns are labeled D and the remaining columns are labeled A.

\begin{theorem}
\label{thm:counting-solutions}
The number of solutions to the sorting match puzzle with parameters $k$ and $\ell$ (i.e., first $k$ rows and first $\ell$ columns labeled A, rest labeled D) is:
\[
\binom{L}{k\ell} \cdot \binom{n^2 - L}{(n-k)\ell} \cdot H(k,\ell) \cdot H(n-k, n-\ell) \cdot H(k, n-\ell) \cdot H(n-k, \ell),
\]
where $$H(a,b) = \frac{\prod_{i=0}^{a-1} i!}{\prod_{i=0}^{a-1} (b+i)!} \cdot (ab)!$$ is the number of standard Young tableaux of rectangular shape $a \times b$, and $$L = k\ell + (n-k)(n-\ell).$$
\end{theorem}

\begin{proof}
By the preceding lemma, numbers $1$ to $k\ell + (n-k)(n-\ell)$ fill $P^{\nwarrow} \cup P^{\searrow}$, while the remaining numbers fill $P^{\nearrow} \cup P^{\swarrow}$. Within each subgrid, the ordering constraints (ascending/descending in rows and columns) make filling equivalent to constructing a standard Young tableau of rectangular shape. The binomial coefficient counts ways to distribute the smallest $k\ell + (n-k)(n-\ell)$ numbers between $P^{\nwarrow}$ and $P^{\searrow}$. Once this choice is made, the hook length formula gives the count for each rectangular subgrid independently.
\end{proof}

\subsection{Unique Solutions}

Note that~\Cref{thm:solvable-puzzle-characterization} characterizes solvability, and gives us conditions for when a permutation grid puzzle has \emph{no} solution. We would now like to characterize when a permutation grid puzzle has \emph{exactly one} solution. Notice that non-uniform grids have multiple solutions as suggested by our counting argument from~\Cref{thm:counting-solutions}: also see~\Cref{fig:multiple-solutions} for examples on puzzles of order two.

\begin{figure}[h]
\centering
\begin{tikzpicture}[scale=0.9]
\begin{scope}
\draw[thick] (0,0) grid (2,2);
\draw[thick] (0,0) rectangle (2,2);
\node at (0.5,1.5) {\large 1};
\node at (1.5,1.5) {\large 3};
\node at (0.5,0.5) {\large 4};
\node at (1.5,0.5) {\large 2};
\node[left] at (-0.2,1.5) {A};
\node[left] at (-0.2,0.5) {D};
\node[above] at (0.5,2.2) {A};
\node[above] at (1.5,2.2) {D};
\node[below] at (1,-0.3) {Solution 1};
\end{scope}

\begin{scope}[xshift=4cm]
\draw[thick] (0,0) grid (2,2);
\draw[thick] (0,0) rectangle (2,2);
\node at (0.5,1.5) {\large 1};
\node at (1.5,1.5) {\large 4};
\node at (0.5,0.5) {\large 3};
\node at (1.5,0.5) {\large 2};
\node[left] at (-0.2,1.5) {A};
\node[left] at (-0.2,0.5) {D};
\node[above] at (0.5,2.2) {A};
\node[above] at (1.5,2.2) {D};
\node[below] at (1,-0.3) {Solution 2};
\end{scope}

\begin{scope}[xshift=9cm]
\draw[thick] (0,0) grid (2,2);
\draw[thick] (0,0) rectangle (2,2);
\node at (0.5,1.5) {\large 4};
\node at (1.5,1.5) {\large 1};
\node at (0.5,0.5) {\large 2};
\node at (1.5,0.5) {\large 3};
\node[left] at (-0.2,1.5) {D};
\node[left] at (-0.2,0.5) {A};
\node[above] at (0.5,2.2) {D};
\node[above] at (1.5,2.2) {A};
\node[below] at (1,-0.3) {Solution 1};
\end{scope}

\begin{scope}[xshift=13cm]
\draw[thick] (0,0) grid (2,2);
\draw[thick] (0,0) rectangle (2,2);
\node at (0.5,1.5) {\large 4};
\node at (1.5,1.5) {\large 2};
\node at (0.5,0.5) {\large 1};
\node at (1.5,0.5) {\large 3};
\node[left] at (-0.2,1.5) {D};
\node[left] at (-0.2,0.5) {A};
\node[above] at (0.5,2.2) {D};
\node[above] at (1.5,2.2) {A};
\node[below] at (1,-0.3) {Solution 2};
\end{scope}
\end{tikzpicture}
\caption{Non-uniform puzzles with multiple solutions. Left pair: the AD/AD puzzle has two valid solutions. Right pair: the DA/DA puzzle also has two valid solutions. In both cases, the non-uniform structure in rows and columns creates incomparable pairs in the poset.}
\label{fig:multiple-solutions}
\end{figure}
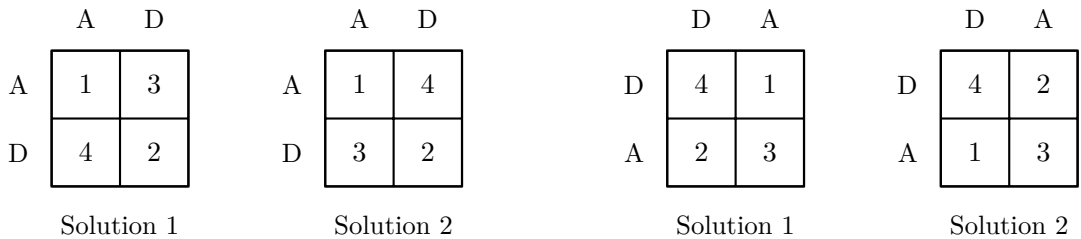

On the other hand, if we fix a row-uniform or column-uniform puzzle, we can characterize when it has a unique solution: this happens precisely when the non-uniform side is alternating, and is filled out by a so-called ``boustrophedon'' pattern (see~\Cref{fig:boustrophedon}).

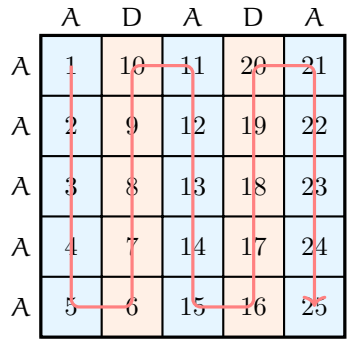
\begin{figure}[h]
\centering
\begin{tikzpicture}[scale=0.8]
    \definecolor{colA}{RGB}{230,245,255}
    \definecolor{colD}{RGB}{255,240,230}

    \fill[colA] (0,0) rectangle (1,5);
    \fill[colD] (1,0) rectangle (2,5);
    \fill[colA] (2,0) rectangle (3,5);
    \fill[colD] (3,0) rectangle (4,5);
    \fill[colA] (4,0) rectangle (5,5);

    \draw[thick] (0,0) grid (5,5);
    \draw[very thick] (0,0) rectangle (5,5);

    \node[above] at (0.5,5) {$A$};
    \node[above] at (1.5,5) {$D$};
    \node[above] at (2.5,5) {$A$};
    \node[above] at (3.5,5) {$D$};
    \node[above] at (4.5,5) {$A$};

    \node[left] at (0,0.5) {$A$};
    \node[left] at (0,1.5) {$A$};
    \node[left] at (0,2.5) {$A$};
    \node[left] at (0,3.5) {$A$};
    \node[left] at (0,4.5) {$A$};

    \node at (0.5,4.5) {1};
    \node at (1.5,4.5) {10};
    \node at (2.5,4.5) {11};
    \node at (3.5,4.5) {20};
    \node at (4.5,4.5) {21};
    \node at (0.5,3.5) {2};
    \node at (1.5,3.5) {9};
    \node at (2.5,3.5) {12};
    \node at (3.5,3.5) {19};
    \node at (4.5,3.5) {22};
    \node at (0.5,2.5) {3};
    \node at (1.5,2.5) {8};
    \node at (2.5,2.5) {13};
    \node at (3.5,2.5) {18};
    \node at (4.5,2.5) {23};
    \node at (0.5,1.5) {4};
    \node at (1.5,1.5) {7};
    \node at (2.5,1.5) {14};
    \node at (3.5,1.5) {17};
    \node at (4.5,1.5) {24};
    \node at (0.5,0.5) {5};
    \node at (1.5,0.5) {6};
    \node at (2.5,0.5) {15};
    \node at (3.5,0.5) {16};
    \node at (4.5,0.5) {25};

    \draw[->, thick, red!50, line width=1.2pt, rounded corners=3pt]
        (0.5,4.5) -- (0.5,3.5) -- (0.5,2.5) -- (0.5,1.5) -- (0.5,0.5) --
        (1.5,0.5) -- (1.5,1.5) -- (1.5,2.5) -- (1.5,3.5) -- (1.5,4.5) --
        (2.5,4.5) -- (2.5,3.5) -- (2.5,2.5) -- (2.5,1.5) -- (2.5,0.5) --
        (3.5,0.5) -- (3.5,1.5) -- (3.5,2.5) -- (3.5,3.5) -- (3.5,4.5) --
        (4.5,4.5) -- (4.5,3.5) -- (4.5,2.5) -- (4.5,1.5) -- (4.5,0.5);
\end{tikzpicture}
\caption{The boustrophedon (snake) pattern for a $5 \times 5$ puzzle with $r = A^5$ and $c = ADADA$. The alternating column constraints force a unique traversal order, shown by the arrows.}
\label{fig:boustrophedon}
\end{figure}

\begin{theorem}
A sorting match puzzle of order $n \geq 2$ has a unique solution if and only if it satisfies one of the following:
\begin{enumerate}
    \item Row-uniform with alternating columns: $r \in \{A^n, D^n\}$ and $c$ alternates between $A$ and $D$.
    \item Column-uniform with alternating rows: $c \in \{A^n, D^n\}$ and $r$ alternates between $A$ and $D$.
\end{enumerate}
\end{theorem}

\begin{proof}
We first show that an alternating constraint pattern on an uniform puzzle has a unique solution. Consider the case $r = A^n$ with $c = (AD)^{n/2}$ (or $(AD)^{(n-1)/2}A$ if $n$ is odd). Every row is ascending left-to-right. Column 1 is ascending (bottom-to-top), column 2 is descending, and so on. This creates a snake-like path: starting from cell $(1,1)$, we traverse up column 1, then across row $n$ to column 2, down column 2, across row 1 to column 3, and so forth. We now argue that each step in this path is \emph{forced} by the constraints, making the solution unique.

We show by induction that value $k$ is at cell $(k, 1)$ for each $k \in \{1, \ldots, n\}$. Since all rows are ascending, the value 1 cannot appear in any column other than the first. Since column 1 is ascending, the value 1 must be in row 1. Thus, cell $(1,1)$ contains 1, this is the base case. For the inductive step, suppose values $1, \ldots, k$ occupy cells $(1, 1), \ldots, (k, 1)$ for some $k < n$. We show that value $k+1$ must be at $(k+1, 1)$.

Suppose $k+1$ is at cell $(i, j)$ with $j \geq 2$. Since row $i$ is ascending, $(i, 1) < (i, j) = k+1$, so $(i, 1) \leq k$. By the inductive hypothesis, the cells containing values $1, \ldots, k$ are exactly $(1,1), \ldots, (k,1)$, hence $i \leq k$. Now consider column 2, which is descending. All cells $(i', 2)$ with $i' > i$ satisfy $(i', 2) < (i, 2) \leq k+1$ (the first inequality by column 2 descending, the second by row $i$ ascending). Thus all such cells need values from $\{1, \ldots, k\}$, but these values are already in column 1. Since there are $n - i \geq n - k \geq 1$ such cells, we have a contradiction.

Therefore $k+1$ is in column 1. Since cells $(1,1), \ldots, (k,1)$ are occupied and column 1 is ascending, we have $(k+1, 1) = k+1$.

After column 1 is filled with $1, \ldots, n$, we show $n+1$ must be at $(n, 2)$. If $n+1$ were at $(i, j)$ with $j \geq 3$, then $(i, 2) < (i, j) = n+1$ by row $i$ ascending. But all values less than $n+1$ are in column 1, so $(i, 2)$ cannot be filled. Thus $n+1$ must be in column 2. A similar inductive argument shows that values $n+1, \ldots, 2n$ occupy cells $(n, 2), (n-1, 2), \ldots, (1, 2)$ respectively: the descending constraint on column 2 forces us to fill it from bottom to top. The pattern alternates: odd columns fill upward, even columns fill downward, creating the boustrophedon pattern. The case $r = D^n$ with alternating columns, and the symmetric cases with $c$ uniform and $r$ alternating, follow by analogous arguments.

Now suppose $r = A^n$ (row-uniform) but $c$ contains two adjacent columns with the same label, say columns $j$ and $j+1$ are both descending. Consider the $2 \times 2$ subgrid formed by rows $i, i+1$ and columns $j, j+1$. Let the four cells be:
\[
\begin{array}{|c|c|}
\hline
w & x \\
\hline
y & z \\
\hline
\end{array}
\]
where $w = (i+1,j)$, $x = (i+1,j+1)$, $y = (i,j)$, $z = (i,j+1)$.

The constraints require  $y < z$ and $w < x$ and $y < w$ and $z < x$. These constraints form a partial order where $y$ is minimum and $x$ is maximum, but $w$ and $z$ are \emph{incomparable} --- neither is constrained to be less than the other. This means there exist valid solutions with $w < z$ and other valid solutions with $z < w$, giving at least two distinct solutions. As a concrete example, consider a $4 \times 4$ puzzle with $r = AAAA$ and $c = DDAD$. The following are both valid solutions:
\[
\renewcommand{\arraystretch}{1.3}
\setlength{\arraycolsep}{6pt}
\begin{array}{c|cccc}
& D & D & A & D \\
\hline
A & 4 & 8 & 9 & 16 \\
A & 3 & 7 & 10 & 15 \\
A & 2 & 6 & 11 & 14 \\
A & 1 & 5 & 12 & 13 \\
\end{array}
\qquad\qquad
\begin{array}{c|cccc}
& D & D & A & D \\
\hline
A & 7 & 8 & 9 & 16 \\
A & 5 & 6 & 10 & 15 \\
A & 3 & 4 & 11 & 14 \\
A & 1 & 2 & 12 & 13 \\
\end{array}
\]

For the mixed case (where neither $r$ nor $c$ is uniform), the hook length formula from the previous section shows the count is always greater than 1.
\end{proof}

\subsection{Nearest Solvable Puzzles}

Imagine Tanvi encounters an unsolvable sorting match puzzle. Unwilling to give up, she plans to change some of the row and column labels to make the puzzle solvable. However, she does not want to be conspicuous in doing so. Hence she wonders: what is the minimum number of label flips needed to make the puzzle solvable?

Given an arbitrary sorting match puzzle $(r, c) \in \{A, D\}^n \times \{A, D\}^n$, we may ask: what is the ``nearest'' solvable puzzle? Using the notation $\|a,b\|$ to denote the Hamming distance between two strings, we define the \emph{nearest solvable puzzle problem}:

\begin{framed}
\noindent\textsc{Nearest Solvable Sorting Match Puzzle} \\[0.5em]
\textbf{Input:} A sorting match puzzle $(r, c) \in \{A, D\}^n \times \{A, D\}^n$ \\
\textbf{Output:} A solvable puzzle $(r^*, c^*)$ minimizing $\|r,r^*\| + \|c,c^*\|$
\end{framed}

Recall from \Cref{thm:solvable-puzzle-characterization} that a puzzle is solvable if and only if it avoids certain forbidden $2 \times 2$ patterns. A key observation is that if either $r$ or $c$ is \emph{uniform} (all $A$'s or all $D$'s), then no forbidden pattern can arise, since such patterns require transitions in \emph{both} rows and columns at adjacent positions. More generally, a puzzle is solvable if both $r$ and $c$ have at most one transition point (from $A$ to $D$ or vice versa), provided these transitions are \emph{compatible} --- both must be $A \to D$ transitions or both must be $D \to A$ transitions.

Notice that by guessing the ``shape'' of the nearest solvable puzzle, we can find the answer in $O(n^2)$ time by brute force. But we can do better. For a string $s \in \{A, D\}^n$, define:
\begin{itemize}
    \item $X(k) = |\{i \leq k : s_i = A\}|$, the number of $A$'s in the first $k$ positions
    \item $\textsf{total}_A(s) = X(n)$, the total number of $A$'s in $s$
\end{itemize}

The cost to transform $s$ into the pattern $A^k D^{n-k}$ (first $k$ positions are $A$, remaining are $D$) is:
\[
\textsf{cost}_{A \to D}(k) = \underbrace{(k - X(k))}_{\text{\# D's to flip in prefix}} + \underbrace{(\textsf{total}_A - X(k))}_{\text{\# A's to flip in suffix}} = k + \textsf{total}_A - 2 \cdot X(k)
\]

Similarly, the cost to transform $s$ into $D^k A^{n-k}$ is:
\[
\textsf{cost}_{D \to A}(k) = X(k) + (n - k - \textsf{total}_A + X(k)) = 2 \cdot X(k) + n - k - \textsf{total}_A
\]

We compute $X(k)$ incrementally via the recurrence $X(0) = 0$ and $X(k) = X(k-1) + [s_k = A]$. This yields recurrences for the costs:
\begin{align*}
\textsf{cost}_{A \to D}(0) &= \textsf{total}_A, & \textsf{cost}_{A \to D}(k) &= \textsf{cost}_{A \to D}(k-1) + \begin{cases} -1 & \text{if } s_k = A \\ +1 & \text{if } s_k = D \end{cases} \\
\textsf{cost}_{D \to A}(0) &= n - \textsf{total}_A, & \textsf{cost}_{D \to A}(k) &= \textsf{cost}_{D \to A}(k-1) + \begin{cases} +1 & \text{if } s_k = A \\ -1 & \text{if } s_k = D \end{cases}
\end{align*}
As we scan from $k = 0$ to $n$, the cost changes by $\pm 1$ at each step, so we can track the running minimum in $O(n)$ time for each string.

\begin{theorem}\label{thm:minimum-flips}
The nearest solvable sorting match puzzle can be computed in $O(n)$ time.
\end{theorem}

\begin{proof}
Let $\textsf{min}_{A \to D}(s)$ and $\textsf{min}_{D \to A}(s)$ denote the minimum costs over all $k$ for the respective patterns. The answer is:
\[
\min\left(
\begin{array}{l}
\min(\textsf{total}_A(r), n - \textsf{total}_A(r)), \\
\min(\textsf{total}_A(c), n - \textsf{total}_A(c)), \\
\textsf{min}_{A \to D}(r) + \textsf{min}_{A \to D}(c), \\
\textsf{min}_{D \to A}(r) + \textsf{min}_{D \to A}(c)
\end{array}
\right).
\]

The first two options make one string uniform (cost = minimum of flipping all to $A$ or all to $D$), leaving the other unchanged. The last two options transform both strings to have at most one compatible transition. Each quantity is computed in a single $O(n)$ pass.
\end{proof}



\section{Permutation Match Puzzles}

In this section, we turn to permutation match puzzles and establish that determining the minimum number of constraint modifications needed to make such a puzzle solvable is NP-complete. Recall that a permutation match puzzle is solvable if and only if its associated constraint graph is acyclic. This motivates the following graph-theoretic problem.

\begin{framed}
\noindent\textsc{Grid Acyclification} \\[0.5em]
\textbf{Input:} A positive integer $n$, a non-negative integer $k$, and a directed graph $G$ on the vertex set $V = \{(i,j) \mid 1 \leq i,j \leq n\}$ such that for each $i \in [n]$, the subgraph induced by the row $R_i := \{(i,j) \mid 1 \leq j \leq n\}$ is an acyclic tournament, and for each $j \in [n]$, the subgraph induced by the column $C_j := \{(i,j) \mid 1 \leq i \leq n\}$ is an acyclic tournament. \\
\textbf{Question:} Can $G$ be made acyclic by deleting at most $k$ vertices?
\end{framed}

Deleting a vertex removes it along with all incident edges. The constraint that each row and column induces an acyclic tournament reflects the structure arising from permutation constraints: each permutation $\pi \in S_n$ defines a total order on $n$ elements, which corresponds to a transitive tournament on the row or column.

The \textsc{Grid Acyclification} problem is closely related to the classical \textsc{Feedback Vertex Set} problem, where one seeks to remove a minimum number of vertices to eliminate all directed cycles. Our variant restricts attention to graphs with the grid structure described above.

\begin{theorem}\label{thm:grid-acyclification-np-complete}
\textsc{Grid Acyclification} is NP-complete.
\end{theorem}

We establish this by a reduction from \textsc{Feedback Vertex Set} on directed graphs. Before presenting the formal construction, we illustrate the key ideas through an example.

\textbf{Example.} Consider the directed graph $H$ on five vertices shown in \Cref{fig:fvs-example}. The graph contains several cycles, including $v_2 \to v_4 \to v_3 \to v_5 \to v_2$ and $v_1 \to v_3 \to v_5 \to v_1$. The minimum feedback vertex set of $H$ has size one, since deleting one vertex ($v_3$ or $v_5$) breaks all cycles.

\begin{figure}[t]
\centering
\begin{tikzpicture}[
    directed/.style={
        blue, 
        thick,
        decoration={
            markings,
            mark=at position 0.55 with {\arrow[scale=1.5]{Stealth}}
        },
        postaction={decorate}
    },
    vertex/.style={circle, fill=black, inner sep=2pt}
]

    \node[vertex, label={above right:$v_1$}] (v1) at (0, 4) {};
    \node[vertex, label={below:$v_2$}] (v2) at (0, 1.8) {};
    \node[vertex, label={below left:$v_4$}] (v4) at (-1.5, 0) {};
    \node[vertex, label={right:$v_5$}] (v5) at (2, 0) {};
    \node[vertex, label={left:$v_3$}] (v3) at (-3.5, 0) {};

    
    \draw[directed] (v2) -- (v1);
    \draw[directed] (v2) -- (v4);
    \draw[directed] (v5) -- (v2);
    \draw[directed] (v5) -- (v4);
    \draw[directed] (v4) -- (v3);

    \draw[directed] (v1) to[bend right=25] (v3);
    \draw[directed] (v5) to[bend right=25] (v1);
    \draw[directed] (v4) to[bend left=15] (v1);
    
    \draw[directed] (v3) to[bend right=70] (v5);

\end{tikzpicture} 
\caption{A directed graph $H$ on five vertices $v_1,v_2,v_3.v_4,v_5$. (Note that all cycles of $H$ pass through $v_3$ and $v_5$; hence, deleting any one of these vertices makes the graph acyclic.) The grid graph $G$ corresponding to $H$ is shown in \Cref{fig:grid-construction}.}\label{fig:fvs-example}
\end{figure}
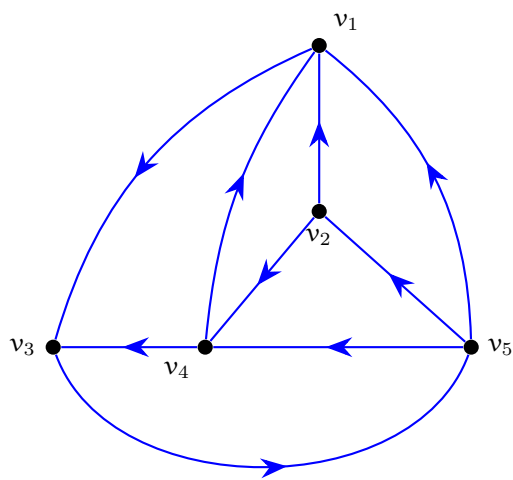

\begin{figure}[h]
    \centering
    \vspace{-0.25cm}
    \includegraphics[scale=0.17]{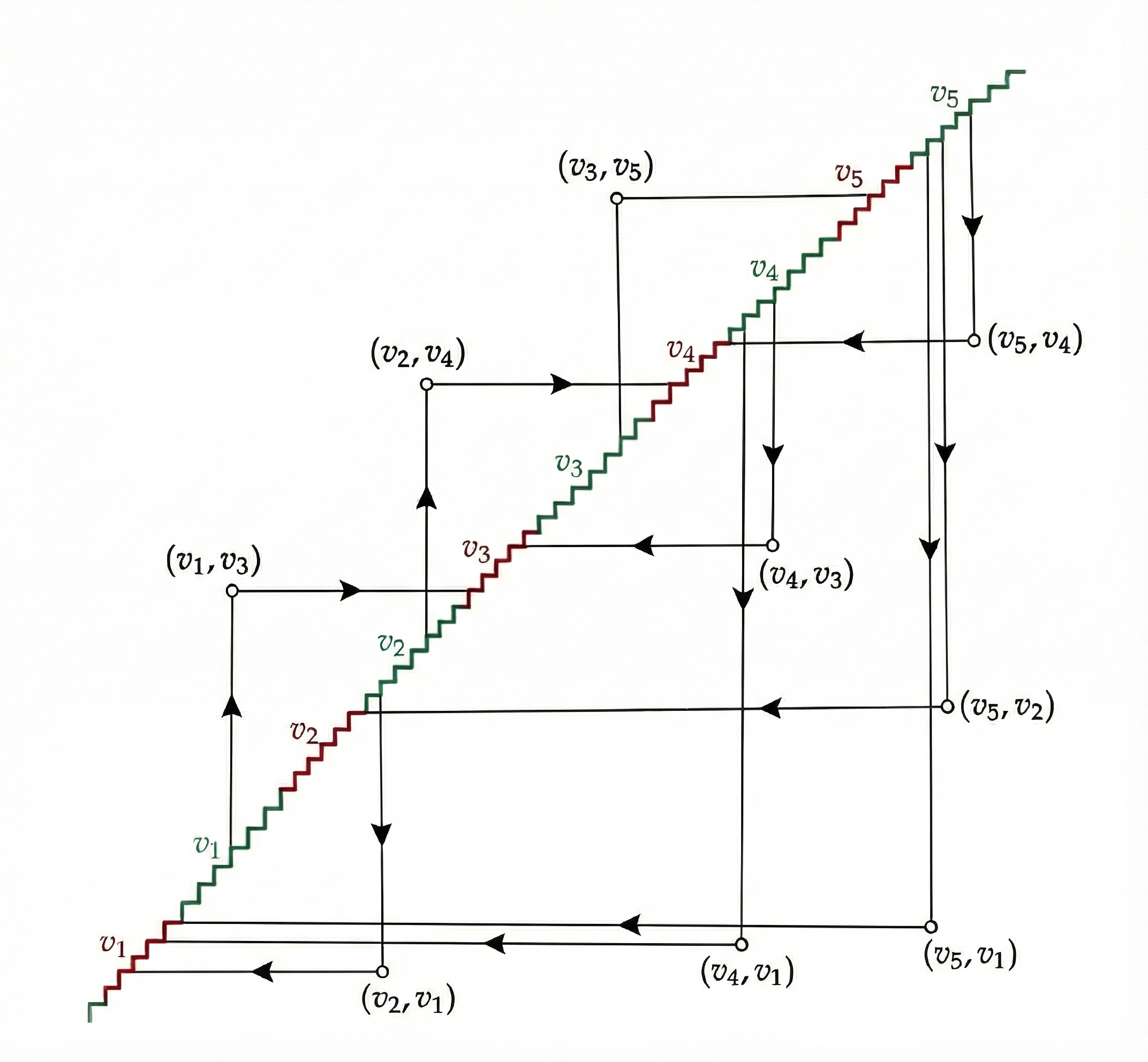}
    \caption{The grid graph $G$ constructed from the directed graph $H$ of \Cref{fig:fvs-example}. Each vertex $v_i$ contributes a red segment (incoming), a black key step $\kappa_i$, and a green segment (outgoing) to the staircase. Edge vertices (the circular dots lying outside the staircase) connect green segments to red segments via orthogonal paths. All other vertices of the grid are inactive (that is, they do not participate in any cycle). (Note that deleting $\kappa_3$ or $\kappa_5$ makes this grid acyclic.)}\label{fig:grid-construction}
\end{figure}

We construct a grid graph $G$ from $H$ as follows. The construction has three types of vertices (see \Cref{fig:grid-construction}):

\begin{enumerate}
    \item \textbf{Staircase vertices:} A diagonal staircase runs from the bottom-left to the top-right corner (edges of the staircase are oriented alternately upwards and rightwards). For each vertex $v_i$ of $H$, the staircase contains a \emph{red segment} (where edges can only enter), followed by a \emph{connector vertex} (shown in black), followed by a \emph{green segment} (where edges can only leave). The connector vertex between the red and green segments of $v_i$ is called the \emph{key vertex} for $v_i$.

    \item \textbf{Edge vertices:} For each edge $(v_i, v_j)$ in $H$, we place an \emph{edge vertex} at a unique grid position.

    \item \textbf{Inactive vertices:} All remaining grid cells contain \emph{inactive vertices}. These are configured as sinks in every row and column tournament, ensuring they never participate in any cycle.
\end{enumerate}

The key insight is that \textbf{each row and column contains at most three active (non-sink) vertices}: at most two from the staircase (where the row/column crosses the diagonal) and at most one edge vertex. This sparse structure allows us to carefully control the edge orientations within each row and column tournament.

A cycle $v_1 \to v_3 \to v_5 \to v_1$ in $H$ corresponds to a cycle in $G$ that traverses:

green $v_1$ $\to$ edge vertex $(v_1, v_3)$ $\to$ red $v_3$ $\to$ (upwards along staircase) $\to$ green $v_3$ $\to$ edge vertex $(v_3, v_5)$ $\to$ red $v_5$ $\to$ (upwards along staircase) $\to$ green $v_5$ $\to$ edge vertex $(v_5, v_1)$ $\to$ red $v_1$ $\to$ (upwards along staircase) $\to$ green $v_1$.

Crucially, deleting the \emph{key vertex} for $v_i$ (the black connector between red $v_i$ and green $v_i$) disconnects all paths through $v_i$, which corresponds to removing $v_i$ from the feedback vertex set in $H$.

\textbf{Formal Construction.} Let $H = (V_H, E_H)$ be a directed graph with $V_H = \{v_1, \ldots, v_n\}$ and $m = |E_H|$ edges. We construct a grid graph $G$ of size $N \times N$ where $N = 4n + m$.

\emph{Staircase structure.} The diagonal of $G$ contains a staircase pattern. For each vertex $v_i \in V_H$, we allocate a segment of the diagonal consisting of:
\begin{itemize}
    \item A \emph{red segment} of $d^-(v_i) + 1$ consecutive diagonal cells, where $d^-(v_i)$ is the in-degree of $v_i$ in $H$.
    \item A single \emph{key vertex} $\kappa_i$.
    \item A \emph{green segment} of $d^+(v_i) + 1$ consecutive diagonal cells, where $d^+(v_i)$ is the out-degree of $v_i$.
\end{itemize}
The segments are arranged along the diagonal in order: red $v_1$, key $\kappa_1$, green $v_1$, red $v_2$, key $\kappa_2$, green $v_2$, and so on.

\emph{Edge vertices.} For each edge $e = (v_i, v_j) \in E_H$, we designate a unique cell $\varepsilon_e$ in the grid. This cell is positioned such that:
\begin{itemize}
    \item It shares a row with exactly one cell from the green segment of $v_i$.
    \item It shares a column with exactly one cell from the red segment of $v_j$.
\end{itemize}
Since each vertex's segment has enough cells to accommodate all incident edges, and each edge uses a distinct cell from each segment, the edge vertices are all distinct and well-defined.

All grid cells not designated as staircase vertices or edge vertices are \emph{inactive vertices}.

\emph{Tournament orientations.} We now specify the edge orientations within each row and column. Let $\mathcal{A}$ denote the set of \emph{active vertices}: all staircase and edge vertices. For each row $R_i$ and column $C_j$:
\begin{itemize}
    \item \textbf{Inactive vertices are sinks:} For every inactive vertex $u$ and every other vertex $w$ in the same row or column, the edge is directed $w \to u$.
    \item \textbf{Staircase edges follow the diagonal:} Within the staircase, edges are directed from lower-left to upper-right: if $(a, a)$ and $(b, b)$ are staircase vertices with $a < b$ in the same row or column, then $(a,a) \to (b,b)$.
    \item \textbf{Edge vertex orientations:} For an edge vertex $\varepsilon_e$ corresponding to edge $(v_i, v_j)$:
    \begin{itemize}
        \item In its row (shared with green $v_i$): the edge goes from the green segment cell to $\varepsilon_e$.
        \item In its column (shared with red $v_j$): the edge goes from $\varepsilon_e$ to the red segment cell.
    \end{itemize}
    \item \textbf{Direct edges between distant active vertices:} If two active vertices share a row or column but are not adjacent in the staircase and neither is an edge vertex connecting them, orient the edge from left-to-right (in rows) or bottom-to-top (in columns).
\end{itemize}

These orientations ensure that each row and column induces an acyclic tournament on its vertices. We now claim that inactive vertices never participate in any cycle of $G$. Indeed, every inactive vertex is a sink with respect to staircase and edge vertices. Further, the subgraph induced by inactive vertices is acyclic by construction. This implies the claim. 

\begin{lemma}\label{lem:cycle-correspondence}
The graph $G$ contains a cycle if and only if $H$ contains a cycle. Moreover, every cycle in $G$ passes through the key vertex $\kappa_i$ for some $v_i \in V_H$.
\end{lemma}

\begin{proof}
$(\Rightarrow)$ Let $C$ be a cycle in $G$. Note that $C$ contains only active vertices. Since staircase edges are directed from lower-left to upper-right, any path along the staircase must eventually leave the staircase (at a green segment) to continue as a cycle. Such a path must use an edge vertex to jump to a later portion of the staircase.

Tracing the cycle $C$: starting from any key vertex $\kappa_i$, the path proceeds through the green segment of $v_i$, exits via some edge vertex $\varepsilon_{(v_i, v_j)}$, enters the red segment of $v_j$, proceeds to $\kappa_j$, and continues. The sequence of key vertices visited corresponds exactly to a cycle in $H$.

$(\Leftarrow)$ Let $v_{i_1} \to v_{i_2} \to \cdots \to v_{i_\ell} \to v_{i_1}$ be a cycle in $H$. We construct a corresponding cycle in $G$: from $\kappa_{i_1}$, follow the staircase through green $v_{i_1}$, take the edge vertex $\varepsilon_{(v_{i_1}, v_{i_2})}$, enter red $v_{i_2}$, proceed to $\kappa_{i_2}$, and continue. This traces a valid directed path that returns to $\kappa_{i_1}$.
\end{proof}

\begin{proof}[Proof of \Cref{thm:grid-acyclification-np-complete}]
Membership in NP is clear: given a set $S$ of at most $k$ vertices, we can verify in polynomial time that $G - S$ is acyclic.

For NP-hardness, we reduce from \textsc{Feedback Vertex Set} on directed graphs, which is NP-complete. Given an instance $(H, k)$ of \textsc{Feedback Vertex Set}, construct the grid graph $G$ as described above. We claim that $H$ has a feedback vertex set of size at most $k$ if and only if $G$ can be made acyclic by deleting at most $k$ vertices.

$(\Rightarrow)$ Suppose $S \subseteq V_H$ is a feedback vertex set of $H$ with $|S| \leq k$. Let $S' = \{\kappa_i : v_i \in S\}$ be the corresponding key vertices in $G$. We claim $G - S'$ is acyclic. By \Cref{lem:cycle-correspondence}, every cycle in $G$ passes through some key vertex $\kappa_i$. If $v_i \in S$, then $\kappa_i \in S'$ is deleted. If $v_i \notin S$, then any cycle through $\kappa_i$ corresponds to a cycle in $H$ through $v_i$, but since $S$ is a feedback vertex set, such a cycle must also pass through some $v_j \in S$, meaning the cycle in $G$ passes through $\kappa_j \in S'$. Thus, all cycles are broken.

$(\Leftarrow)$ Suppose $S'$ is a set of at most $k$ vertices in $G$ such that $G - S'$ is acyclic. We construct a feedback vertex set $S$ for $H$ with $|S| \leq k$ using a \emph{pushing argument}.

First, recall that we may assume $S'$ contains no inactive vertices (removing an inactive vertex from $S'$ cannot create new cycles).

For each vertex $u \in S'$, we ``push'' it to a nearby key vertex:
\begin{itemize}
    \item If $u = \kappa_i$ is already a key vertex, include $v_i$ in $S$.
    \item If $u$ is in the red segment of $v_i$, include $v_i$ in $S$. (Any cycle through $u$ must continue to $\kappa_i$, so deleting $\kappa_i$ also breaks this cycle.)
    \item If $u$ is in the green segment of $v_i$, include $v_i$ in $S$. (Any cycle through $u$ must have come from $\kappa_i$.)
    \item If $u = \varepsilon_{(v_i, v_j)}$ is an edge vertex, include $v_i$ in $S$. (Any cycle through $u$ must pass through either $\kappa_i$ or $\kappa_j$; we choose $v_i$.)
\end{itemize}

This maps each vertex in $S'$ to a vertex in $S$ (possibly with repetitions), so $|S| \leq |S'| \leq k$. We claim $S$ is a feedback vertex set. Suppose for contradiction that $H - S$ contains a cycle $v_{i_1} \to \cdots \to v_{i_\ell} \to v_{i_1}$. By \Cref{lem:cycle-correspondence}, $G$ contains a corresponding cycle through $\kappa_{i_1}, \ldots, \kappa_{i_\ell}$. Since $v_{i_j} \notin S$ for all $j$, no key vertex $\kappa_{i_j}$ is the image of any vertex in $S'$ under our pushing map. But then the original set $S'$ did not break this cycle in $G$, contradicting that $G - S'$ is acyclic.

The construction of $G$ from $H$ is polynomial in $|V_H| + |E_H|$, completing the reduction.
\end{proof}

\section{Concluding Remarks}

We have presented a comprehensive study of permutation match puzzles, a family of constraint satisfaction problems on grids inspired by a physical puzzle. Our main contributions include a complete characterization of solvability via the acyclicity of an associated constraint graph, which reduces to an elegant ``at most one switch'' condition on the row and column labels. When solutions exist, we derived an exact counting formula using hook lengths, connecting these puzzles to the rich combinatorics of standard Young tableaux. For unsolvable instances, we provided a linear-time algorithm to compute the minimum number of label flips required to restore solvability. These results extend naturally to three-dimensional grids and partially filled puzzles, where the solvability criterion generalizes in a clean manner.

Several directions remain open for further investigation. First, while we established NP-completeness for the generalized setting where rows and columns specify arbitrary permutations, it would be interesting to identify tractable special cases --- for instance, when the permutations come from a restricted family or when the constraint graph has bounded treewidth. Second, our results focus on complete grids; understanding the complexity of partial filling problems where only a subset of cells must be filled could lead to connections with constraint propagation and arc consistency. Third, the hook length formula suggests deeper algebraic structure; exploring whether representation-theoretic techniques can yield further insights into the solution space would be worthwhile. The interplay between simple local constraints and global structure makes permutation match puzzles a fertile ground for both recreational mathematics and algorithmic exploration.



\bibliography{lipics-v2021-sample-article}

\end{document}